\documentclass[a4paper,twocolumn,11pt,accepted=2021-10-13]{quantumarticle}
\pdfoutput=1
\usepackage[utf8]{inputenc}
\usepackage[english]{babel}
\usepackage[T1]{fontenc}
\usepackage{amsmath}
\usepackage{hyperref}

\usepackage{float}

\usepackage{tikz}
\usepackage{lipsum}

\usepackage{amsmath}  			
\usepackage{amsfonts}  			
\usepackage{graphicx}   			
\usepackage{mathrsfs, amsthm, amssymb}
\usepackage{bbm, bm}

\usepackage{nicefrac}    			

\usepackage{fourier-orns}  		

\usepackage[braket, qm]{qcircuit}	

\usepackage{nicefrac} 

\usepackage{tikz}
\usetikzlibrary{calc,decorations.pathmorphing,shapes}

\newcounter{sarrow}

\newcommand{\indep}{\perp \!\!\! \perp}

\newtheorem{theorem}{\textsf{Theorem}}
\newtheorem{definition}{\textsf{Definition}}

\newtheorem{corollary}{\textsf{Corollary}} 
\newtheorem*{example}{\textsf{Example}}

\newtheorem{Rule}{\underline{\textsf{Rule}}} 

\usepackage[sort&compress,numbers]{natbib}

\usepackage{contour}
\usepackage{ulem}

\contourlength{0.8pt}

\newcommand{\myuline}[1]{%
  \uline{\phantom{#1}}%
  \llap{\contour{white}{#1}}%
}
\renewcommand{\emph}{\textit}

\begin{document}

\title{On safe post-selection for Bell tests with ideal detectors:\newline Causal diagram approach}

\author{Pawel Blasiak}
\affiliation{Institute of Nuclear Physics Polish Academy of Sciences, PL-31342 Krak\'ow, Poland}
\orcid{0000-0002-3457-2870}
\email{pawel.blasiak@ifj.edu.pl}
\author{Ewa Borsuk}
\affiliation{Institute of Nuclear Physics Polish Academy of Sciences, PL-31342 Krak\'ow, Poland}
\orcid{0000-0003-0446-2782}
\email{ewa.borsuk@ifj.edu.pl}
\author{Marcin Markiewicz}
\affiliation{International Centre for Theory of Quantum Technologies, University of Gda\'nsk, PL-80308 Gda\'nsk, Poland}
\orcid{0000-0002-8983-9077}
\email{marcinm495@gmail.com}

\maketitle

\begin{abstract}
\onecolumn Reasoning about Bell nonlocality from the correlations observed in post-selected data is always a matter of concern. This is because conditioning on the outcomes is a source of non-causal correlations, known as a \textit{selection bias}, rising doubts whether the conclusion concerns the actual causal process or maybe it is just an effect  of processing the data. Yet, even in the idealised case without detection inefficiencies, post-selection is an integral part of experimental designs, not least because it is a part of the entanglement generation process itself. In this paper we discuss a broad class of scenarios with post-selection on multiple spatially distributed outcomes. A simple criterion is worked out, called the \textit{all-but-one} principle, showing when the conclusions about nonlocality from breaking Bell inequalities with post-selected data remain in force. Generality of this result, attained by adopting the high-level diagrammatic tools of causal inference, provides safe grounds for systematic reasoning based on the standard form  of multipartite Bell inequalities in a wide array of entanglement generation schemes, without worrying about the dangers of selection bias. In particular, it can be applied to post-selection defined by single-particle events in each detection chanel when the number of particles in the system is conserved.
\end{abstract}
\vspace{0.3cm}
\twocolumn

\section{Introduction}
The study of experimental correlations provides a window into the underlying causal mechanisms, even when their exact nature remains obscured. In his seminal works~\cite{Be93}, John~Bell showed that seemingly innocuous assumptions about the causal structure of realistic models leave a mark on the observed statistics. The conclusion has been that the violation of certain inequalities is incompatible with the assumption of \textit{locality} and \textit{free choice} (or \textit{measurement independence}). Surprisingly, such violations can systematically occur in quantum theory, potentially undermining our dearly held assumptions about how nature works. Given how troubling this conclusion might be, it is hardly surprising how thoroughly Bell’s result has been scrutinised in the last few decades, both theoretically~\cite{Wi14a,BrCaPiScWe14,Sc19} and experimentally~\cite{As15}. Its importance is acknowledged by the term \textit{'Bell nonlocality'} which refers to experimental situations demonstrating inconsistency of the observed correlations with causal (or realist) explanations maintaining both assumptions at the same time. It is believed to be a source of quantum advantage in the communication~\cite{EkRe14} and information  tasks~\cite{NiCh00}.

A simplified picture of a Bell experiment consists of a series of measurements made by  space-time separated parties on systems prepared in some entangled state. 
However, it comes with a challenge as to the straight conclusion regarding Bell nonlocality when it comes to the analysis of real experimental designs. An important issue concerns the presence of post-selection in the data collection process. One source of the problem lies with the measurement part of the experiment in which some of the events are missed out  due to the inefficiencies of real detectors. It is called the detection loophole~\cite{Gi14c,La14a} and will not be addressed in this paper. In the following we shall assume ideal detectors and thus focus on post-selection due to the preparation part of the experiment. 

It is often the case that post-selection in the experiment is due to the specifics of the entanglement generation process itself. Typically this boils down to the occurrence of a certain pattern in the outcomes deemed to be interesting for the purpose at hand, i.e., exhibiting entanglement~\cite{PaChLuWeZeZu12,ErKrZe20,WaScLaTh20}. Some popular techniques of this sort include: heralding by some other event (cf., event-ready detection~\cite{ZuZeHoEk93}), time-bin entanglement~\cite{Fr89} or selecting single-particle detections in each experimental channel (cf., recent proposals in Refs.~\cite{ZhBaLuZhYaRuPa08,KrHoLaZe17,KyErHoKrZe20,BlMa19,KiPrChYaHaLeKa18,KiChLiHa20,BlBoMaKi21,StLiMoTu17,BeLoCo17,CaBeCoLo19,ZhLiLiPeSuHuHe18,WaQiDiChChYoHe19,LoCo18,NoCaCoLo20,SuWaLiXuXuLiGu20,BaChPrLiChHuKi20} or a variety of integrated photonic implementations~\cite{WaScLaTh20}). So, the generic structure of events is richer than that required for the intended Bell inequality, and post-selection aims at retaining only those experimental trials, based on some well-defined criterion, which are potentially interesting for the violation of the desired inequality. 
This poses an issue regarding the legitimacy of the conclusion about Bell nonlocality in such scenarios, since conditioning is often a source of non-causal correlations. In the field of causal inference the problem is known as a \textit{selection bias} or \textit{Berkson’s paradox}~\cite{Pe09,SpGlSc00,PeGlJe16,PeMa18}. The difficulty being that, in the presence of post-selection, it might be conditioning that leads to correlations breaking Bell inequalities without necessarily claiming Bell nonlocality. How critical it is for the analysis of Bell experiments may attest the effort to close the detection loophole~\cite{La14a}, which exploits post-selection due to detector inefficiencies. 
Here we will assume ideal detectors and focus only on post-selection due to entanglement generation.

It is interesting to ask about general conditions 
\begin{center}
\textit{\parbox{0.9\columnwidth}{\textsf{when post-selection, due to  entanglement generation, does not compromise the conclusion of nonlocality from the violation of some given Bell inequality.}}}
\end{center}
So far, this problem has been  discussed only for some particular scenarios for two and three parties, and the analysis typically involved the entire pattern of experimental outcomes present in a given experiment~\cite{PoHaZu97,Zu00,ScVaCaMa11}. In some cases, those issues can be overcome by certain modifications in the experimental arrangements (like for time-bin entanglement in Refs.~\cite{AeKwLaZu99,LiVaChCaMa10,CaCaSaCuFuToFi15,VeAgToAvLaVaVi18}). However, apart from those particular cases no attempt has been made at a general analysis of Bell nonlocality in the presence of post-selection due to the specifics of entanglement generation process. We note that  comprehensiveness of such an analysis would require a discussion of both assumptions, \textit{locality} and \textit{free choice}, underlying the derivation of Bell inequalities.

In this paper we give a general criterion for any multipartite scenario, called the \textit{all-but-one} principle, that can be expressed by the following simple intuition:
\begin{center}
\textit{\parbox{0.9\columnwidth}{\textsf{if post-selection can be resolved with one party excluded, then it is safe for Bell nonlocality arguments.}}}
\end{center}
Meaning that, in such a case, the reasoning based on the standard Bell inequalities is justified despite the issues of post-selection. Crucially, the generality of the result owes to the high-level diagrammatic tools of causal inference honed by Judea Pearl~\cite{Pe09,SpGlSc00,PeGlJe16,PeMa18}. We give a full poof of this criterion preceded with a brief discussion of the selection bias and Bell nonlocality under post-selection. 

\section{Selection bias and d-separation rules}

Post-selection is a procedure of  
rejecting some of the data from the analysis of an experiment. Technically, it boils down to estimating experimental probabilities subject to some additional condition which depends on the outcomes. 
It is crucial to make a warning that post-selection is not a harmless procedure, since it is often a source of additional correlations in the retained data. This is potentially dangerous for the task of identification of causal relationships between the variables from the observed correlations. In the field of causal inference the problem is known as a \textit{selection bias} or \textit{Berkson's paradox}. 

Let us illustrate the problem with a simple example due to Elwert \& Winship~\cite{PeMa18}. Consider three features of Hollywood actors who could be beautiful $B$, talented $T$, and some of them make it to be celebrities $C$. We may reasonably expect that beauty $B$ and talent $T$ contribute to an actor being considered a celebrity $C$ (pushed to the extreme, imagine that one of these features is enough to become a celebrity), but in general population beauty $B$ and talent $T$ are completely unrelated to one another. Suppose that this is the whole story and hence Fig.~\ref{Fig_SelectionBias} (on the left) illustrates the causal diagram behind the  data. 
Now, if we focus on the subpopulation of those actors who made it to the status of celebrities $C$, then correlation between beauty $B$ and talent $T$ appears (despite the fact that they were independent to begin with). See Fig.~\ref{Fig_SelectionBias} (on the right). Clearly, seeing an unattractive celebrity makes it more likely that the person is a talented actor.  And vice versa, celebrities who are bad actors are more often found to be good looking. (Note that pushed to the extreme, this inverse relation may even become a certain conclusion). These correlations are \textit{non-causal}, i.e., they arise merely due to conditioning or restricting the data generated by the causal diagram in Fig.~\ref{Fig_SelectionBias} (on the left). This example illustrates the warning against careless attribution of causal origin to correlations in the post-selected data. 

\begin{figure}
\centering
\includegraphics[width=0.9\columnwidth]{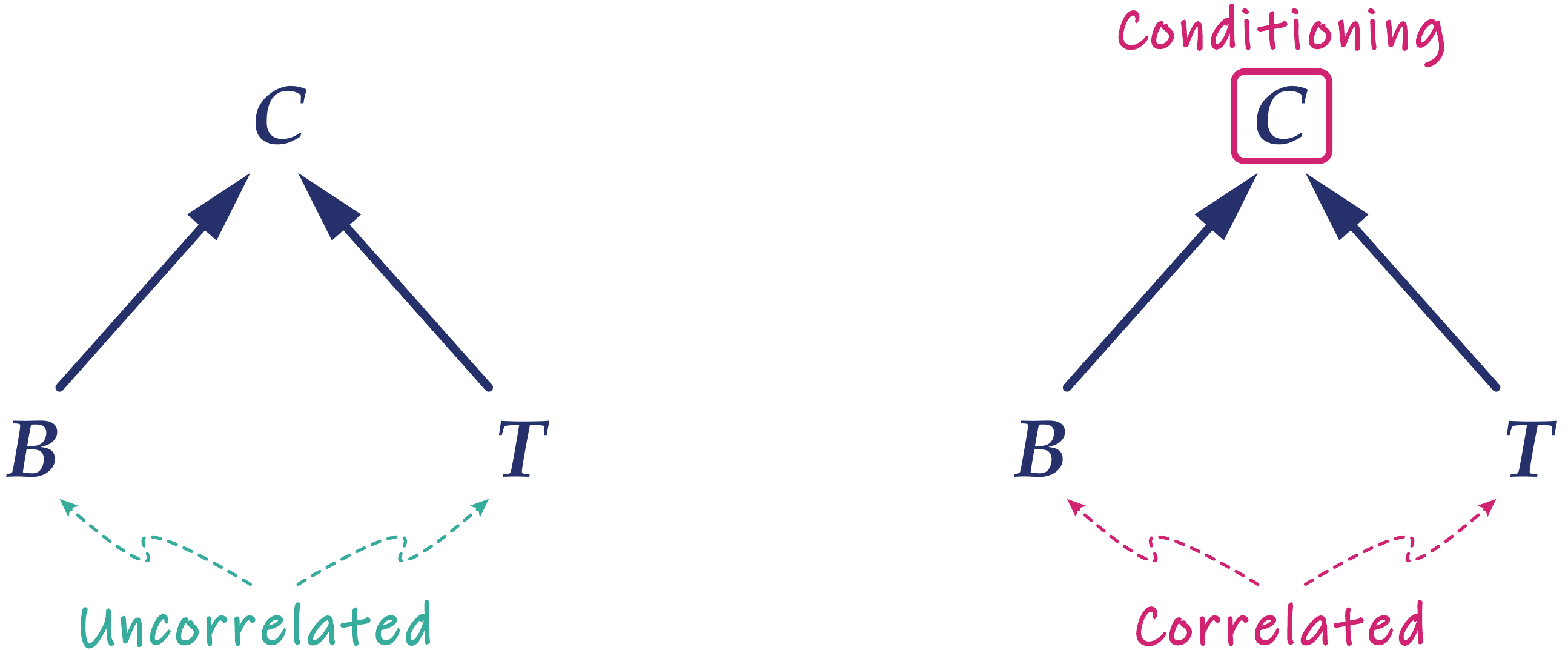}
\caption{\label{Fig_SelectionBias}{\bf\textsf{\mbox{Selection bias.}}} Consider three variables $B$ (beauty), $T$ (talent) and $C$ (celebrity) modelled by a causal diagram on the left. 
It follows that the variables $B$ and $T$ are independent, since the only path joining the variables is blocked by a \textit{collider} $C$ (\textbf{\textsf{Rule~\ref{1}}}). 
However, on the right, \textit{conditioning} on the collider $C$, depicted by a red box, unblocks the path between $B$ and $T$ making the variables likely dependent  (\textbf{\textsf{Rule~\ref{3}}}).}
\end{figure}

Interestingly, the pattern of independencies between the variables can be deduced from the structure of the causal diagram itself. It has been shown to boil down to the so called {\textbf{d-separation}} criterion, see~\cite{Pe09,SpGlSc00,PeMa18,PeGlJe16}. In a nutshell, the idea consists of inspecting all paths in the causal diagram connecting two variables of interest:
\begin{center}\textit{\parbox{0.91\columnwidth}{\textsf{if all those paths are \textit{blocked} then the variables are necessarily independent (otherwise the variables are likely dependent).}}}
\end{center}
The concept of \textit{blocking} a path is defined by the following three simple \textit{$d$-separation} rules (see Fig.~\ref{Fig_d-separation} for illustration):

\begin{Rule}\label{1}
\normalfont{A path is blocked if there is a \textit{collider} along the way, that is  a node with pair of arrows on the path that collide head-to-head.}\vspace{-0.1cm}
\end{Rule}

\begin{Rule}\label{2}
\normalfont{Conditioning on a \textit{non-collider} blocks the path (where non-collider is a node along the way with pair of arrows meeting head-to-tail or tail-to-tail).}
\end{Rule}\vspace{-0.1cm}

\begin{Rule}\label{3}
\normalfont{Conditioning on a \textit{collider} (or its descendant) removes the block from \textbf{\textsf{Rule~\ref{1}}}.}
\end{Rule}

In the following analysis, those rules will provide systematic insight into the pattern of conditional independences arising from specific post-selection procedures.\vspace{0.2cm}

\begin{figure}
\centering
\includegraphics[width=0.97\columnwidth]{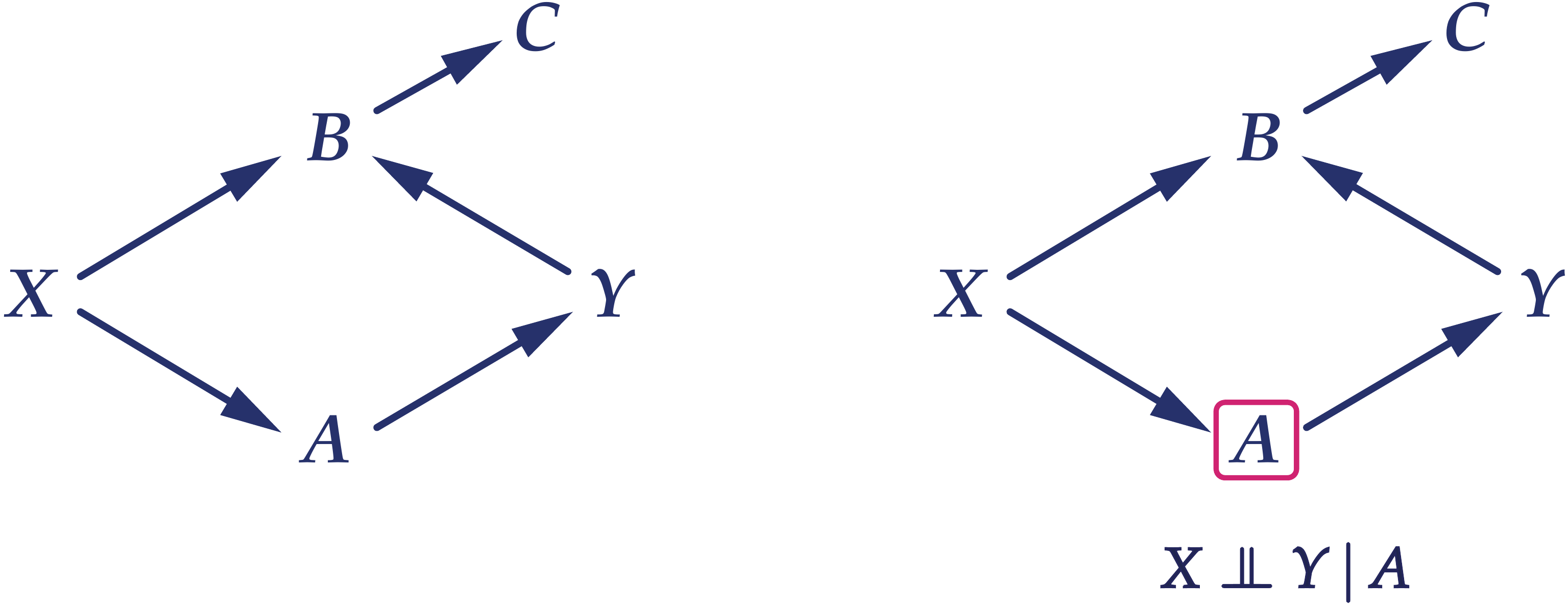}
\caption{\label{Fig_d-separation}{\bf\textsf{\mbox{Illustration of the \textit{d-separation} rules.}}} In the causal diagram nodes $X$ and $Y$ are connected via two paths $X\rightarrow B\leftarrow Y$ and $X\rightarrow A\rightarrow Y$. The path $X\rightarrow B\leftarrow Y$ is blocked since $B$ is a collider by \textbf{\textsf{Rule~\ref{1}}} and conditioning on $B$ or $C$ unblocks the path by \textbf{\textsf{Rule~\ref{3}}}. Whereas path $X\rightarrow A\rightarrow Y$ is unblocked since $A$ is a non-collider and conditioning on $A$ blocks the path by \textbf{\textsf{Rule~\ref{2}}}. If both paths are blocked then $X$ and $Y$ are said to be \textit{d-separated} and then the variables are independent. In this example it happens only in the case of conditioning on $A$ which is depicted on the right by a red box.}

\end{figure}

\section{Bell nonlocality and three causal assumptions}

In the following, we consider the usual Bell-type scenario with several parties $A,B,C,...$ conducting experiments in space-time separated regions. The whole experiment consists of a series of trials in which each party chooses a setting $x,y,z,...$ and makes a measurement registering the outcome $a,b,c,...$\,. For further convenience, let us denote the set of possible outcomes by $\mathcal{O}$, i.e., we have respectively $\mathcal{O}_{\scriptscriptstyle A}$, $\mathcal{O}_{\scriptscriptstyle B}$, $\mathcal{O}_{\scriptscriptstyle C}$, ...\,. Then, after many repetitions, the parties compare their results calculating the statistics given by the set of distributions $P_{\scriptscriptstyle abc...|xyz...}$ which describe the probability of observing outcomes $a,b,c,...$ given measurements $x,y,z,...$ were made. For conciseness, following the terminology in Ref.~\cite{Sc19}, we call such obtained set of probabilities $\mathcal{P}\equiv\{P_{\scriptscriptstyle abc...|xyz...}\}_{\scriptscriptstyle xyz...}$ a \textit{"behaviour"}. Note that  all probabilities in the behaviour $\mathcal{P}$ are supposed to be calculated \textit{without} rejecting any trial from the experiment (no post-selection is made).

\begin{figure*}[t]
\centering
\includegraphics[width=1.9\columnwidth]{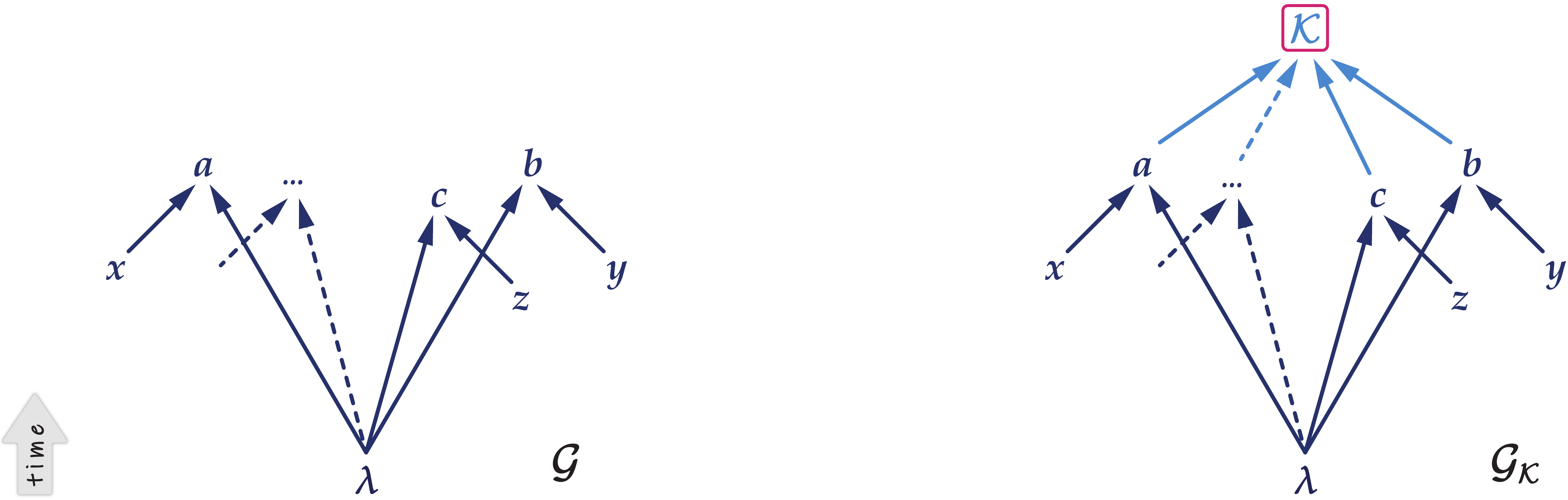}
\caption{\label{Fig_CausalStructure}{\bf\textsf{\mbox{Causal structure in a Bell  experiment.}}} One the left, the graph $\mathcal{G}$ describes causal relations between variables  $a,b,c,...$ (measurement outcomes), $x,y,z,...$ (choice of settings), and some hidden variable $\lambda$. It generates the behaviour $\mathcal{P}\equiv\{P_{\scriptscriptstyle abc...|xyz...}\}_{\scriptscriptstyle xyz...}$. 
On the right, the graph $\mathcal{G}_{\scriptscriptstyle \mathcal{K}}$ incorporates post-selection into the experiment by conditioning (red box) on the additional outcome-dependent variable $\mathcal{K}(a,b,c,...)$. The latter introduces potential \textit{selection bias} into such obtained behaviour $\mathcal{P}_{\scriptscriptstyle \mathcal{K}}\equiv\{P_{\scriptscriptstyle abc...|xyz...\,\mathcal{K}}\}_{\scriptscriptstyle xyz...}$.}
\end{figure*}

Bell inequalities are algebraic constraints in the form
\begin{eqnarray}\label{Bell-inequality}
\mathcal{I}(\mathcal{P})&\equiv&\sum_{\scriptscriptstyle abc...\atop xyz...}\,s_{\scriptscriptstyle xyz...}^{\scriptscriptstyle abc...}\,P_{\scriptscriptstyle abc...|xyz...}\ \leqslant\ I_{\scriptscriptstyle L}\,,
\end{eqnarray}
where $s_{\scriptscriptstyle xyz...}^{\scriptscriptstyle abc...}$ and $I_{\scriptscriptstyle L}$ are some numbers. These inequalities are derived under three assumptions called realism, locality and free choice. The \textit{realism assumption} posits that the observed correlations can be explained by the causal influence between the variables relevant for the experiment, that is measurement outcomes $a,b,c,...$ and settings $x,y,z,...$, as well as some unobserved (hidden) variables collectively denoted by $\lambda$. Thus, by conditioning on \textit{a priori} unknown $\lambda$, we can always write~\cite{Be93,Wi14a,BrCaPiScWe14,Sc19}
\begin{eqnarray}\label{realism}
P_{\scriptscriptstyle abc...|xyz...}&=&\sum_{\scriptscriptstyle \lambda}\ P_{\scriptscriptstyle abc...|xyz...\,\lambda}\cdot P_{\scriptscriptstyle \lambda|xyz...}\,.
\end{eqnarray}
Then, by invoking spatio-temporal structure of the experiment certain conditional independencies between the variables can be justified. First, the variables in different space locations cannot affect each other and the causal influences propagate respecting temporal order of events. Second, the hidden variable $\lambda$ is identified to be in the common past of variables representing the outcomes $a,b,c,...$\,, but not the variables representing choice of the settings $x,y,z,...$ and hence the latter cannot be affected by $\lambda$. The ensuing causal structure of the variables modelling the experiment is depicted in the causal graph $\mathcal{G}$ in Fig.~\ref{Fig_CausalStructure} (on the left). This readily translates into conditional independencies in the statistics generated by those causal models.\footnote{Note that in this argument causal relationships are considered as prior to the statistical relations, with the latter derivable from an appropriate structural causal model compatible with a given causal structure. This is the leitmotif of the causal inference field~\cite{Pe09,SpGlSc00,PeGlJe16,PeMa18}.} They are referred to as the \textit{locality} assumption
\begin{eqnarray}\label{locality}
P_{\scriptscriptstyle abc...|xyz...\,\lambda}&=&P_{\scriptscriptstyle a|x\lambda}\cdot P_{\scriptscriptstyle b|y\lambda}\cdot P_{\scriptscriptstyle c|z\lambda}\cdot...\,,
\end{eqnarray}
and the \textit{free choice} assumption (also called the \textit{measurement independence} assumption)
\begin{eqnarray}\label{free-choice}
P_{\scriptscriptstyle \lambda|xyz...}&=&P_{\scriptscriptstyle \lambda}\,.
\end{eqnarray}
Within the causal model framework, these relations are a straightforward application of the $d$-separation rules to the diagram in Fig.~\ref{Fig_CausalStructure} (on the left). [Eq.~(\ref{locality}) follows by iterative use of \textbf{\textsf{Rule~\ref{2}}} given conditioning on non-collider node $\lambda$, and Eq.~(\ref{free-choice}) is an application of \textbf{\textsf{Rule~\ref{1}}} to colliders $a,b,c,...$\,; cf. proof of \textbf{\textsf{Theorem}~\ref{theorem}}.]

To summarise, each Bell inequality Eq.~(\ref{Bell-inequality}) is a simple algebraic consequence of the three assumptions in Eqs.~(\ref{realism})-(\ref{free-choice}). It means that the violation of some Bell inequality entails the impossibility of explaining the observed behaviour $\mathcal{P}$ in a causal model maintaining both locality and free choice at the same time. The essence of Bell's theorem is to point out situations in which  quantum theory predicts violation of those inequalities~\cite{Be93,Wi14a,BrCaPiScWe14,Sc19}.

\section{Post-selection issues}

Crucially, the statistics used for estimation of probabilities in the behaviour $\mathcal{P}$ should include every experimental trial for a valid conclusion from breaking Bell inequality in Eq.~(\ref{Bell-inequality}) to be drawn. In practice, however, some sort of post-selection is always made. Let us formalise this concept by assuming that the causal structure encoded in the causal graph $\mathcal{G}$ remains the same, but admits a richer variety of outcomes $\widetilde{\mathcal{O}}$, i.e., we have $\widetilde{\mathcal{O}}_{\scriptscriptstyle A}\supset\mathcal{O}_{\scriptscriptstyle A}$, $\widetilde{\mathcal{O}}_{\scriptscriptstyle B}\supset\mathcal{O}_{\scriptscriptstyle B}$, $\widetilde{\mathcal{O}}_{\scriptscriptstyle C}\supset{\mathcal{O}}_{\scriptscriptstyle C}$, ...\,. It means that, in addition to the outcomes of interest $\mathcal{O}$, the experiment predicts results which will have to be rejected in the analysis. Then post-selection boils down to conditioning on some outcome-dependent variable
\begin{eqnarray}
\mathcal{K}&\equiv&\mathcal{K}(a,b,c,...)\,.
\end{eqnarray}
Say, for $\mathcal{K}=1$ we accept the result, otherwise for $\mathcal{K}=0$ the result is rejected. This procedure aims at recovering the proper structure of outcomes  for the  intended Bell inequality Eq.~(\ref{Bell-inequality}). In other words, the reduction $\widetilde{\mathcal{O}}\leadsto\mathcal{O}$ is achieved by making sure that the unwanted results drop out under the conditioning, i.e., $P_{\scriptscriptstyle abc...|xyz...\,\lambda\mathcal{K}}=0$ if $a\notin\mathcal{O}_{\scriptscriptstyle A}$ or $b\notin\mathcal{O}_{\scriptscriptstyle B}$ or $c\notin\mathcal{O}_{\scriptscriptstyle C}$, ...\,. Note that the value of $\mathcal{K}$ is decided only after the parties meet to compare their results. Hence the causal graph takes the form $\mathcal{G}_{\scriptscriptstyle\mathcal{K}}$ in Fig.~\ref{Fig_CausalStructure} (on the right).

In this way we get a new behaviour $\mathcal{P}_{\scriptscriptstyle \mathcal{K}}\equiv\{P_{\scriptscriptstyle abc...|xyz...\,\mathcal{K}}\}_{\scriptscriptstyle xyz...}$ which looks like a good candidate for a test of Bell inequalities Eq.~(\ref{Bell-inequality}). Indeed, all premises seem to be satisfied, i.e., the 'right' causal graph $\mathcal{G}$ with the appropriate pattern of outcomes $\mathcal{O}$, except one detail: there is conditioning in such obtained statistics. 
This raises worries as regards the validity of the conclusions, reached by using $\mathcal{P}_{\scriptscriptstyle \mathcal{K}}$ in Eq.~(\ref{Bell-inequality}), as a legitimate proof of Bell nonlocality. The selection bias may serve here as a warning of how easily post-selection can lead to false causal conclusions, cf. Fig.~\ref{Fig_SelectionBias}. In the case of a Bell test it might happen that conditioning (post-selection) bootstraps the correlations so that $\mathcal{I}(\mathcal{P}_{\scriptscriptstyle \mathcal K})> I_{\scriptscriptstyle L}$, while for the full statistics it remains $\mathcal{I}(\mathcal{P})\leqslant I_{\scriptscriptstyle L}$ in agreement with the causal graph $\mathcal{G}$ in Fig.~\ref{Fig_CausalStructure} (i.e., with locality and free choice maintained). This is possible because conditioning ruins the independence structure of Eqs.~(\ref{locality})~and~(\ref{free-choice}) which is required to prove Eq.~(\ref{Bell-inequality}). To see this, note that the outcomes $a,b,c,...$ play the role of  colliders in the causal graph $\mathcal{G}$ but conditioning on their descendent $\mathcal{K}$ in graph $\mathcal{G}_{\scriptscriptstyle\mathcal{K}}$ in Fig.~\ref{Fig_CausalStructure} opens paths that were previously blocked, cf. \textbf{\textsf{Rule~\ref{1}}}~and~\textbf{\textsf{Rule~\ref{3}}}, thereby introducing correlations into the data which may fake Bell inequalities. How critical it might be for the analysis of Bell nonlocality may attest the effort to close the detection loophole~\cite{BrCaPiScWe14,Sc19,Gi14c,La14a} (which is a case of post-selection too).

Having warned against jumping to hasty conclusions with post-selected data, it is then natural to ask: 
\begin{center}
\textit{\parbox{0.9\columnwidth}{\textsf{When is it possible to make a conclusive Bell argument in the post-selected regime?}}}
\end{center}
In order to make it precise, we assume that the considered causal structure for the experiment is given by the diagram $\mathcal{G}{\scriptscriptstyle \mathcal{K}}$ in Fig.~\ref{Fig_CausalStructure} and make the following definition:
\begin{definition}[Safe post-selection]\label{Safe-post-selection}\ \\Post-selection procedure specified by the variable $\mathcal{K}(a,b,c,...)$ is considered to be \myuline{safe} if the locality and free choice assumptions still hold in the post-selected regime, i.e.,
\begin{eqnarray}\label{locality-K}
P_{\scriptscriptstyle abc...|xyz...\,\lambda\mathcal{K}}&=&P_{\scriptscriptstyle a|x\lambda\mathcal{K}}\cdot P_{\scriptscriptstyle b|y\lambda\mathcal{K}}\cdot P_{\scriptscriptstyle c|z\lambda\mathcal{K}}\cdot...\,,
\end{eqnarray}
and
\begin{eqnarray}\label{free-choice-K}
P_{\scriptscriptstyle \lambda|xyz...\,\mathcal{K}}&=&P_{\scriptscriptstyle \lambda|\mathcal{K}}\,.
\end{eqnarray}
\end{definition}
 An immediate consequence is the observation:
\begin{corollary}\label{Corollary-Safe-post-selection}
If post-selection $\mathcal{K}$ is safe, then in the post-selected regime the same set of Bell inequalities holds, i.e.,  \begin{eqnarray}
\mathcal{I}(\mathcal{P})\ \leqslant\ I_{\scriptscriptstyle L}&\ \Rightarrow\ &\mathcal{I}(\mathcal{P}_{\scriptscriptstyle \mathcal{K}})\ \leqslant\ I_{\scriptscriptstyle L}\,.
\end{eqnarray}
\end{corollary}
\noindent It follows from the fact that each Bell inequality is  obtained by algebraic manipulation of the expression on the l.h.s. of Eq.~(\ref{Bell-inequality}) assuming Eqs.~(\ref{locality}) and (\ref{free-choice}) hold. Clearly, the same must be true for $\mathcal{P}_{\scriptscriptstyle \mathcal{K}}$ since the same algebra, now with Eqs.~(\ref{locality-K}) and (\ref{free-choice-K}), must give the same result.

In this way, the problem of validity of reasoning in the post-selected regime with the same Bell inequalities is phrased in terms of conditional independencies in a given causal structure. Given post-selection $\mathcal{K}$, the latter can be efficiently scrutinised with the diagrammatic tools of causal inference ($d$-separation rules) applied to the causal graph $\mathcal{G}_{\scriptscriptstyle\mathcal{K}}$.

\section{Main result}

\begin{figure*}[t]
\center
\includegraphics[width=1.6\columnwidth]{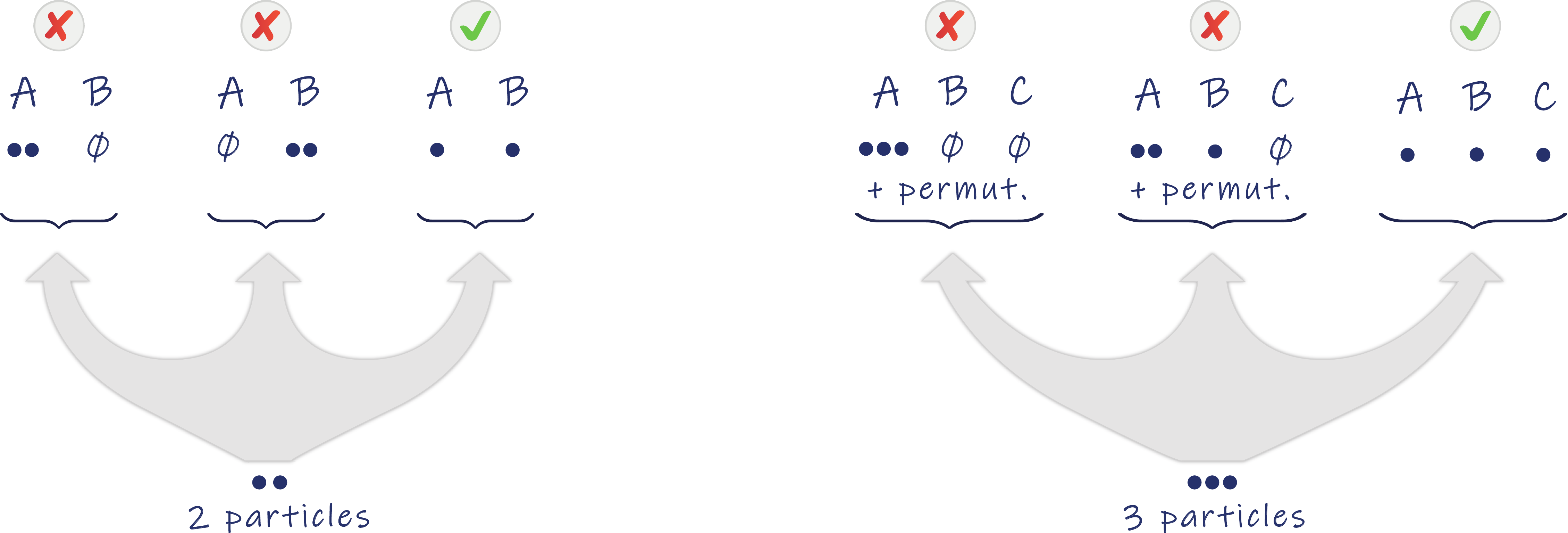}
\caption{\label{Fig_all-but-one}{\bf\textsf{\mbox{Illustration of the \textit{all-but-one} principle.}}} Suppose that in a repeated experiment two particles arrive to Alice (A) and Bob (B) in all possible configurations depicted on the left. Let the meaningful measurements are only those performed on a single particle by each party. If the number of particles is conserved, then Alice alone will know which trials to discard (and similarly for Bob). On the right, there is another party Charlie (C) and three particles are distributed in all possible ways. Again, conservation of particles allows to decide by any two of them $AB$, $AC$ or $BC$ which trials need to be rejected (without knowing what is happening respectively in $C$, $B$ or $A$).}
\end{figure*}

Suppose that the behaviour $\mathcal{P}$ produced by causal graph $\mathcal{G}$ features extra correlations due to the specifics of the preparation procedure (e.g. the number of particles in the experiment is conserved). This often serves as a means of simplification for post-selection $\mathcal{K}$. Let us consider the following property:


\begin{definition}[All-but-one]\label{all-but-one}\ \\Post-selection $\mathcal{K}$ conforms to the \myuline{all-but-one} principle, if it can be fully determined without knowing one of the outcomes. Formally, this boils down to a condition on the form of the variable $\mathcal{K}=\mathcal{K}(a_{\scriptscriptstyle 1},a_{\scriptscriptstyle 2},...,a_{\scriptscriptstyle N})$ which requires that it can be reduced to 
\begin{eqnarray}
\mathcal{K}\,=\,\mathcal{K}(a_{\scriptscriptstyle 1},...\,,\!\not\!{a}_{\scriptscriptstyle k},...\,,a_{\scriptscriptstyle N})\,,\label{all-but-one-EQ}
\end{eqnarray}
for each $k=1,...\,,N$. Here, the symbol $\not\!{a}_{\scriptscriptstyle k}$ means that the outcome for $k$-th party is missing.
\end{definition}
Informally, this means that in a Bell scenario all parties except one is always enough to know whether post-selection ends with a success or not ($\mathcal{K}=1\ \text{or}\ 0$). For example, for three parties $A$, $B$ and $C$, already two of them may decide if post-selection takes place or not, i.e., we have
\begin{eqnarray}\label{all-but-one-EQ-three}
\mathcal{K}(a,b,c)\,=\,\mathcal{K}(b,c)\,=\,\mathcal{K}(a,c)\,=\,\mathcal{K}(a,b)\,.\ \ \ \ 
\end{eqnarray}

\begin{example}
\normalfont{A typical situation where \textit{all-but-one} principle can be readily applied is when the number of particles is conserved. Suppose that $N$ particles are distributed among $N$ parties which receive the particles in different configurations. Let the interesting measurement results are only those when there is a single particle per party. This means that, on top of the valid experimental runs there will be trials in which some parties will register no or more than one particle. Hence the experiment must resort to  post-selection which consists of retaining only those trials when each party reports a single particle on their side. Observe that since the total number of particles $N$ is conserved, such a post-selection conforms to the \textit{all-but-one} principle. This is because gathering the outcomes from $N-1$ parties is enough to infer the number of particles received by the missing one (i.e., $N-1$ parties registering a single particle may conclude that the remaining one registers a single particle too, since the total number of particles is $N$) and hence to resolve post-selection only by themselves. See Fig.~\ref{Fig_all-but-one} for an illustration and Refs.~\cite{PaChLuWeZeZu12,KrHoLaZe17,WaScLaTh20,ZhBaLuZhYaRuPa08,KyErHoKrZe20,ErKrZe20,BlMa19,KiPrChYaHaLeKa18,KiChLiHa20,BlBoMaKi21,StLiMoTu17,BeLoCo17,CaBeCoLo19,ZhLiLiPeSuHuHe18,LoCo18,WaQiDiChChYoHe19,NoCaCoLo20,SuWaLiXuXuLiGu20,BaChPrLiChHuKi20,WaQiDiChChYoHe19} for some experimental designs with post-selection of the \textit{all-but-one} type.}
\end{example}


Let us observe that the \textit{all-but-one} principle has non-trivial consequences for the causal graph $\mathcal{G}_{\scriptscriptstyle \mathcal{K}}$ in Fig.~\ref{Fig_CausalStructure}. Namely, if the statistics generated by graph $\mathcal{G}$ gives a promise of \textbf{\textsf{Definition~\ref{all-but-one}}}, then one of the arrows pointing to $\mathcal{K}$ in graph $\mathcal{G}_{\scriptscriptstyle \mathcal{K}}$ can be always erased without in any way affecting the generated statistics (in particular, this means that $\mathcal{P}_{\scriptscriptstyle \mathcal{K}}$ will remain unchanged).

Now we can state our main result:

\begin{figure*}[t]
\centering
\includegraphics[width=2.07\columnwidth]{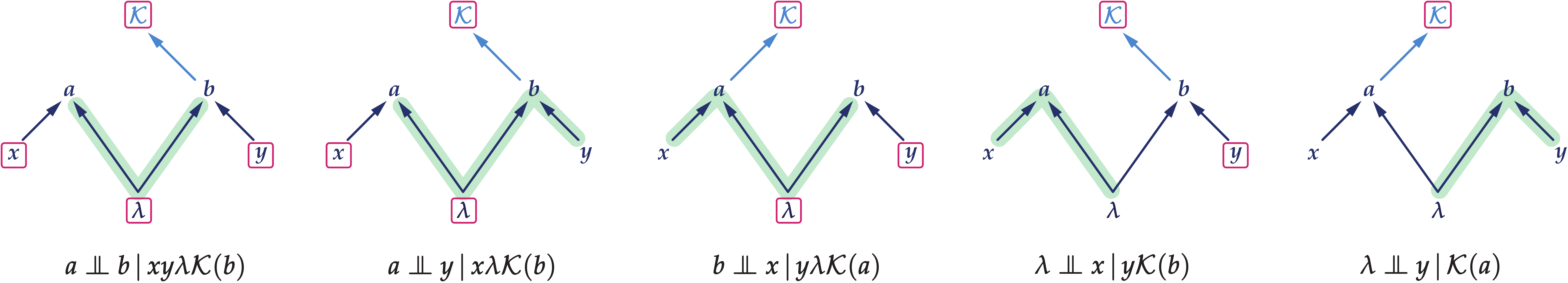}
\caption{\label{Fig_Proof}{\bf\textsf{\mbox{Graphical proof of \textbf{Theorem~\ref{theorem}} (two parties).}}} Each graph illustrates structure of conditioning in Eqs.~(\ref{locality-proof})\,-\,(\ref{free-choice-proof}) depicted by red boxes around the variables. Marked in green are paths joining variables of interest for which the respective independencies are inferred by the $d$-separation \textbf{\textsf{Rules~\ref{1}\,-\,\ref{3}}}. In the first three graphs (on the left) conditioning on the non-collider $\lambda$ blocks the paths, while for the last two graphs (on the right) the paths are blocked by the colliders $a$ and $b$ respectively. Note that without erasing one of the arrows coming to $\mathcal{K}$, as allowed by condition Eq.~(\ref{all-but-one-EQ-two}),  the inference of conditional independencies would not be possible.}
\end{figure*}

\begin{theorem}\label{theorem}\ \\For arbitrary number of parties, post-selection which conforms to the \myuline{all-but-one} principle is always \myuline{safe}.
\end{theorem}
\begin{proof}
Here we sketch the proof for two parties $A$ and $B$ which serves to illustrate the main ideas. For the full proof see \textsf{Appendix}. We need to justify that  Eqs.~(\ref{locality-K}) and (\ref{free-choice-K}) hold under the following condition:
\begin{eqnarray}
\label{all-but-one-EQ-two}
\mathcal{K}\,=\,\mathcal{K}(b)\,=\,\mathcal{K}(a)\,,
\end{eqnarray}
i.e., the \textit{all-but-one} principle in \textbf{\textsf{Definition~\ref{all-but-one}}}.

As for Eq.~(\ref{locality-K}) the reasoning follows the usual route starting with the standard chain rule, i.e., we have
\begin{eqnarray}\nonumber
P_{\scriptscriptstyle ab|xy\lambda\mathcal{K}}&=&P_{\scriptscriptstyle a|bxy\lambda\mathcal{K}}\cdot P_{\scriptscriptstyle b|xy\lambda\mathcal{K}}\,.
\end{eqnarray}
Then the proof boils down to justification of the following conditional independencies:
\begin{eqnarray}\nonumber
\begin{array}{llllll}
P_{\scriptscriptstyle a|bxy\lambda\mathcal{K}}&=&P_{\scriptscriptstyle a|xy\lambda\mathcal{K}}&&\text{since}&a\indep  b\,|\,xy\lambda\mathcal{K}(b)\\[3pt]
&=&P_{\scriptscriptstyle a|x\lambda\mathcal{K}}&&\text{since}&a\indep  y\,|\,x\lambda\mathcal{K}(b)\,,\\[3pt]
P_{\scriptscriptstyle b|xy\lambda\mathcal{K}}&=&P_{\scriptscriptstyle b|y\lambda\mathcal{K}}&&\text{since}&b\indep  x\,|\,y\lambda\mathcal{K}(a)\,.
\end{array}\\\label{locality-proof}
\end{eqnarray}
Each of them can be inferred from the causal graph $\mathcal{G}_{\scriptscriptstyle{K}}$ and application of the $d$-separation \textbf{\textsf{Rule~\ref{2}}} to the non-collider node $\lambda$. See  Fig.~\ref{Fig_Proof} for illustration.

Similarly, we can justify  Eq.~(\ref{free-choice-K}) and get
\begin{eqnarray}\nonumber
\begin{array}{llllll}
P_{\scriptscriptstyle \lambda|xy\mathcal{K}}&=&P_{\scriptscriptstyle \lambda|y\mathcal{K}}&&\text{since}&\lambda\indep  x\,|\,y\mathcal{K}(b)\\[3pt]
&=&P_{\scriptscriptstyle \lambda|\mathcal{K}}&& \text{since}&\lambda\indep  y\,|\,\mathcal{K}(a)\,.
\end{array}\\
\label{free-choice-proof}
\end{eqnarray}
This time it follows from the $d$-separation \textbf{\textsf{Rule~\ref{1}}}, since $a$ and $b$ are colliders respectively, and \textbf{\textsf{Rule~\ref{3}}} does not apply (note that in neither case $\mathcal{K}$ is a descendent). See  Fig.~\ref{Fig_Proof} for illustration.

Note that crucial for this line of reasoning is the flexibility of the expression $\mathcal{K}$ in Eq.~(\ref{all-but-one-EQ-two}), due to the \textit{all-but-one} principle, which permits to erase one of the arrows coming to $\mathcal{K}$. This trick prevents unblocking certain paths required for the inference of conditional independences. [For example, in last graph on the right in Fig.~\ref{Fig_Proof} retaining arrow $b\rightarrow\mathcal{K}$ would have opened the (green) path $\lambda\rightarrow b\leftarrow x$, since $\mathcal{K}$ becomes then a descendent of the collider $b$ and \textbf{\textsf{Rule~\ref{3}}} applies.]

In general, the complexity of paths that need to be considered in the causal graph $\mathcal{G}_{\scriptscriptstyle \mathcal{K}}$ grows with the number of parties and then the $d$-separation tools prove indispensable for this kind of analysis, see \textsf{Appendix}.
\end{proof}

\section{Discussion}

Because entanglement is not a property generated on demand, every Bell experiment must resort to post-selection. However, this opens the doors to the selection bias introducing non-causal correlations into the data, and thus threatening the conclusions expected to be drawn from the experiment. Therefore it is important that the analysis of Bell nonlocality takes this fact into account. In this paper we gave a simple criterion, called the \textit{all-but-one} principle, that allows for safe reasoning in the post-selected regime. Technically, we prove a theorem showing that Bell inequalities derived from the full causal graph which includes  conditioning due to post-selection of the \textit{all-but-one} type remain unchanged. It means that the conclusions drawn from breaking Bell inequalities with such a
post-selected data remain in full force. Beyond the foundational research and application in multipartite entanglement generation schemes~\cite{PaChLuWeZeZu12,ErKrZe20,WaScLaTh20}, this criterion should be significant for quantum cryptography and device independent certification~\cite{Sc19,EkRe14}.

Novelty of the result reported in this work is three-fold: {(\textit{a})} it concerns \textit{any} multipartite scenario with an arbitrary number of outcomes and settings, {(\textit{b})} it pertains to \textit{any} Bell inequality that can be derived in a given scenario, and {(\textit{c})} \textit{both} assumptions of locality and free choice are explicitly discussed in our analysis. The generality of the \textit{all-but-one} principle should be compared with other treatments of post-selection problem due to entanglement generation~\cite{PoHaZu97,Zu00,ScVaCaMa11}.

Let us emphasise that the \textit{all-but-one} principle draws on a special kind of correlations built into the data due to the specifics of the preparation procedure. Typically, if events of interest consist of arrival of a single particle in each detection channel and the number of particles is known and conserved, then such a post-selection fulfils the \textit{all-but-one} principle. We remark that it is a common situation in quantum optical schemes for entanglement generation, see Refs.~\cite{PaChLuWeZeZu12,ErKrZe20,WaScLaTh20}. Some recent proposals based on coincidence counts for high-dimensional multi-particle entanglement that fall within the \textit{all-but-one} principle include entanglement by path identity~\cite{KrHoLaZe17,KyErHoKrZe20}, entanglement without touching~\cite{BlMa19,KiPrChYaHaLeKa18,KiChLiHa20,BlBoMaKi21} or spatial overlap of indistinguishable particles~\cite{BeLoCo17,CaBeCoLo19}. See also Refs.~\cite{ZhLiLiPeSuHuHe18,WaQiDiChChYoHe19,LoCo18,NoCaCoLo20,SuWaLiXuXuLiGu20,BaChPrLiChHuKi20}. For completeness, we note that the principle  is not applicable to time-bin entanglement scheme~\cite{Fr89} which requires specific treatment~\cite{AeKwLaZu99,LiVaChCaMa10,CaCaSaCuFuToFi15,VeAgToAvLaVaVi18}.

Note also that the \textit{all-but-one} principle has limitations. Although applicable in many theoretical settings it does not hold in situations with detector inefficiencies (when the number of particles is not predictable). This is a serious matter of concern for experimental tests of Bell inequalities leading to the so called detection loophole which has to be analysed by other means~\cite{BrCaPiScWe14,Sc19,Gi14c,La14a}. We also note that an important experimental technique based on event-ready-detection~\cite{ZuZeHoEk93} is beyond the scope of the principle in the present form (however, it allows for a straightforward extension to include that scenario too). 

We remark that in the paper we take a conservative approach to the analysis of Bell nonlocality with causes propagating forward-in-time. For a discussion of retrocausality see e.g.~\cite{Pr96,WhAr20}. Note also that we consider a situation in which both assumptions of locality and free choice are maintained at the same time. For a discussion of partial relaxation of those assumptions see~\cite{BlPoYeGaBo21} and references therein.

Finally, let us highlight the role of conceptual tools of causal inference~\cite{Pe09,SpGlSc00,PeGlJe16,PeMa18} in the present analysis. Not only this is an inspiring and rigorous framework for the discussion of correlations vs cause-and-effect relations, but comes equipped with the high-level diagrammatic tools ($d$-separation rules) which prove indispensable for the treatment of multipartite Bell scenarios with many observers and outcomes. Despite a fairly recent development of the field of causal inference outside of physics, those methods have already successfully influenced the research in quantum foundations, see e.g.~\cite{WoSp15,ChKuBrGr15,RiAgVeJaSpRe15,RiGiChCoWhFe16,AlBaHoLeSp17,ChLePi18,ChCaAgDiAoGiSc18,Ca18,BlPoYeGaBo21,BlBo21}.

\section*{Acknowledgments}
We thank M. \.Zukowski for bringing post-selection issues to our attention. We appreciate Y.-S. Kim and R. Lo Franco for helpful comments. We acknowledge partial support by the Foundation for Polish Science (IRAP project, ICTQT, contract no. MAB/2018/5, co-financed by EU within Smart Growth Operational Programme).

\bibliographystyle{apsrev4-1custom}
\bibliography{CombQuant}

\onecolumn\newpage
\appendix

\section*{\textbf{Appendix}}

Here we prove \textbf{\textsf{Theorem~\ref{theorem}}} from the main text. For sake of illustration, we start with the  case of three parties. Then we give the full proof for any number of parties.

\section*{$\bullet$ \,Proof of \textbf{\textsf{Theorem~\ref{theorem}}} for \textit{three} parties A, B and C}

\textit{[The following proof for three parties aims to better illustrate some additional aspects which do not arise in the two-party case. It also serves to emphasise the significance of $d$-separation tools of causal inference~\cite{Pe09,SpGlSc00,PeGlJe16,PeMa18} for this kind of analysis.]}

\begin{proof}Let us consider a Bell experiment with three parties described by the causal graph $\mathcal{G}_{\mathcal{K}}$ in Fig.~\ref{Fig_CausalStructure_ABC} and assume that post-selection conforms to the \textit{all-but-one} principle in \textbf{\textsf{Definition~\ref{all-but-one}}}, i.e., we have 
\begin{eqnarray}\label{all-but-one-EQ-three-SI}
\mathcal{K}(a,b,c)\,=\,\mathcal{K}(b,c)\,=\,\mathcal{K}(a,c)\,=\,\mathcal{K}(a,b)\,.
\end{eqnarray}
Crucially, this property allows erasing one of the three arrows coming to $\mathcal{K}$ without  affecting the generated statistics. This trick will be used to infer conditional independencies in the post-selected behaviour $\mathcal{P}_{\scriptscriptstyle \mathcal{K}}$.

\begin{figure}[h]
\centering
\includegraphics[width=1\columnwidth]{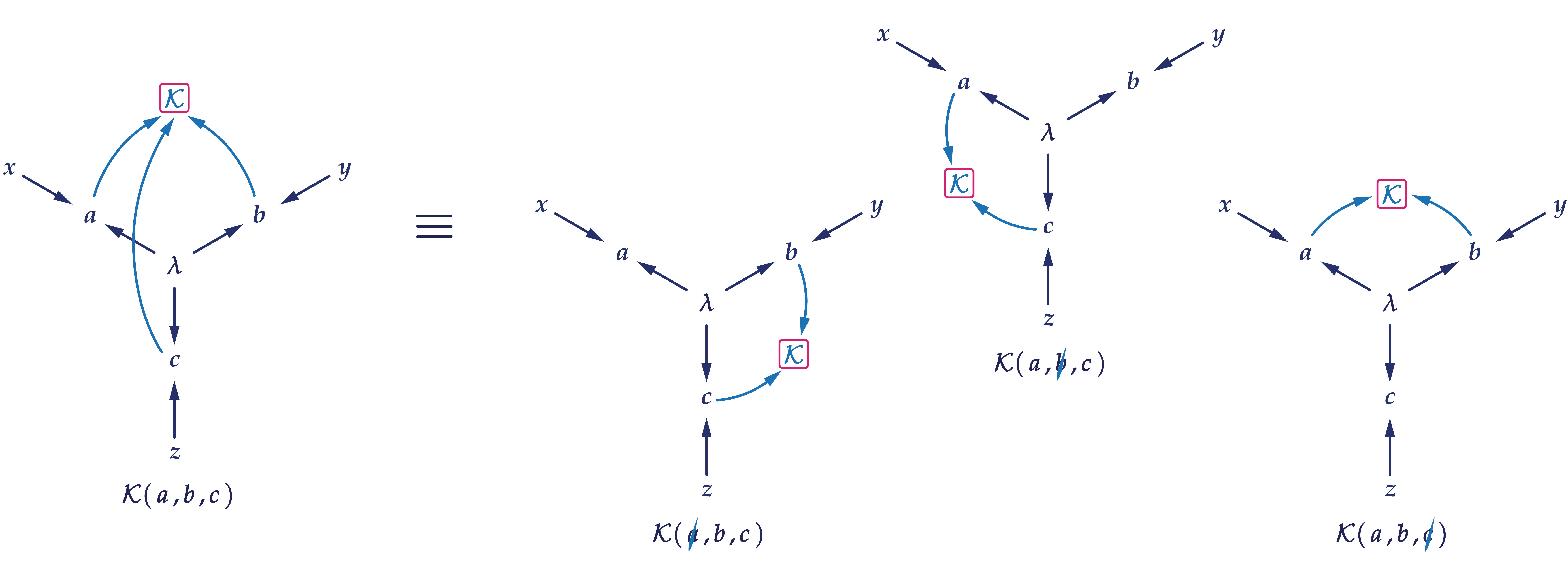}
\caption{\label{Fig_CausalStructure_ABC}{\bf\textsf{\mbox{Causal structure in a Bell experiment for \textit{three} parties with post-selection.}}} On the left, unfolded on a plain is the causal graph $\mathcal{G}_{\scriptscriptstyle \mathcal{K}}$ describing causal relations between variables in a Bell experiment ($a,b,c$ - outcomes, $x,y,z$ - measurement settings, and $\lambda$ - hidden variable). Variable $\mathcal{K}$ represents post-selection (where the red box means conditioning).  Cf. Fig.~\ref{Fig_CausalStructure}.
On the right, three causal graphs which are equivalent to graph $\mathcal{G}_{\scriptscriptstyle \mathcal{K}}$, if the \textit{all-but-one} principle in Eq.~(\ref{all-but-one-EQ-three-SI}) holds (then one of the three arrows coming to $\mathcal{K}$ can be always erased without affecting the generated statistics).}
\end{figure}

In order to prove \textbf{\textsf{Theorem~\ref{theorem}}} we need to justify that both Eqs.~(\ref{locality-K}) and (\ref{free-choice-K}) in \textbf{\textsf{Definition~\ref{Safe-post-selection}}} hold under the condition in Eq.~(\ref{all-but-one-EQ-three-SI}). 

As a straightforward application of the chain rule we get
\begin{eqnarray}
P_{ abc|xyz\lambda\mathcal{K}}\ =\ P_{ a|bcxyz\lambda\mathcal{K}}\,\cdot\, P_{ bc|xyz\lambda\mathcal{K}}\ =\ P_{ a|bcxyz\lambda\mathcal{K}}\,\cdot\, P_{ b|cxyz\lambda\mathcal{K}}\,\cdot\, P_{ c|xyz\lambda\mathcal{K}}\,.
\end{eqnarray}
Now, the prove of Eq.~(\ref{locality-K}) boils down to showing the following sequence of conditional independencies:
\begin{eqnarray}\label{locality-proof-three-1}
\begin{array}{llllll}
P_{ a|bcxyz\lambda\mathcal{K}}&=&P_{ a|cxyz\lambda\mathcal{K}}&&\text{since}&a\indep  b\,|\,cxyz\lambda\mathcal{K}(b,c)\\[3pt]
&=&P_{ a|xyz\lambda\mathcal{K}}&\qquad&\text{since}&a\indep  c\,|\,xyz\lambda\mathcal{K}(b,c)\\[3pt]
&=&P_{ a|xz\lambda\mathcal{K}}&&\text{since}&a\indep  y\,|\,xz\lambda\mathcal{K}(b,c)\\[3pt]
&=&P_{ a|x\lambda\mathcal{K}}&&\text{since}&a\indep z\,|\,x\lambda\mathcal{K}(b,c)\,.
\end{array}\ \ 
\end{eqnarray}
\begin{eqnarray}\label{locality-proof-three-2}
\begin{array}{llllll}
P_{ b|cxyz\lambda\mathcal{K}}&=&P_{ b|xyz\lambda\mathcal{K}}&\qquad&\text{since}&b\indep  c\,|\,xyz\lambda\mathcal{K}(a,c)\\[3pt]
&=&P_{ b|yz\lambda\mathcal{K}}&&\text{since}&b\indep  x\,|\,yz\lambda\mathcal{K}(a,c)\\[3pt]
&=&P_{ b|y\lambda\mathcal{K}}&&\text{since}&b\indep z\,|\,y\lambda\mathcal{K}(a,c)\,.
\end{array}\ \ \ \ 
\end{eqnarray}
\begin{eqnarray}\label{locality-proof-three-3}
\begin{array}{llllll}
P_{ c|xyz\lambda\mathcal{K}}&=&P_{ c|yz\lambda\mathcal{K}}&\qquad&\text{since}&c\indep  x\,|\,yz\lambda\mathcal{K}(a,b)\\[3pt]
&=&P_{ c|z\lambda\mathcal{K}}&&\text{since}&c\indep  y\,|\,z\lambda\mathcal{K}(a,b)\,.
\end{array}\ \ \ \ \ \ \ 
\end{eqnarray}
All of them can be inferred by inspecting paths joining variables in question in the causal graph $\mathcal{G}_{\scriptscriptstyle \mathcal{K}}$ and the use of $d$-separation rules. See Fig.~\ref{Fig_Proof_ABC}. In each case there are two possible paths which are blocked by conditioning on the non-collider node $\lambda$ (\textbf{\textsf{Rule~\ref{2}}}). Note that in order to get the required conditional independencies all conditions in Eq.~(\ref{all-but-one-EQ-three-SI}) need to be used, i.e., the lack of the respective arrow coming to $\mathcal{K}$ in each of the graphs in Fig.~\ref{Fig_Proof_ABC} is essential.

In a similar manner we can prove Eq.~(\ref{free-choice-K}), that is we can prove the following conditional independencies:
\begin{eqnarray}\label{free-choice-proof-three}
\begin{array}{llllll}
P_{ \lambda|xyz\lambda\mathcal{K}}&=&P_{ \lambda|yz\mathcal{K}}&\qquad&\text{since}& \lambda\indep  x\,|\,yz\mathcal{K}(b,c)\\[3pt]
&=&P_{ \lambda|z\mathcal{K}}&&\text{since}& \lambda\indep  y\,|\,z\mathcal{K}(a,c)\\[3pt]
&=&P_{ \lambda|\mathcal{K}}&&\text{since}& \lambda\indep z\,|\,\mathcal{K}(a,b)\,.
\end{array}\ \ \ \ 
\end{eqnarray}
Now there is always one path joining the relevant variables, and in each case it is blocked by the respective collider $a$, $b$ and $c$ (\textbf{\textsf{Rule~\ref{1}}}). See Fig.~\ref{Fig_Proof_ABC}. Notice that, here as well, having all three arrows coming to $\mathcal{K}$ would spoil the proof, since it would lift the block from the respective collider by conditioning on its descendent $\mathcal{K}$ (\textbf{\textsf{Rule~\ref{3}}}).\\
\end{proof}

\newpage

\begin{figure}[H]
\centering
\includegraphics[width=1\columnwidth]{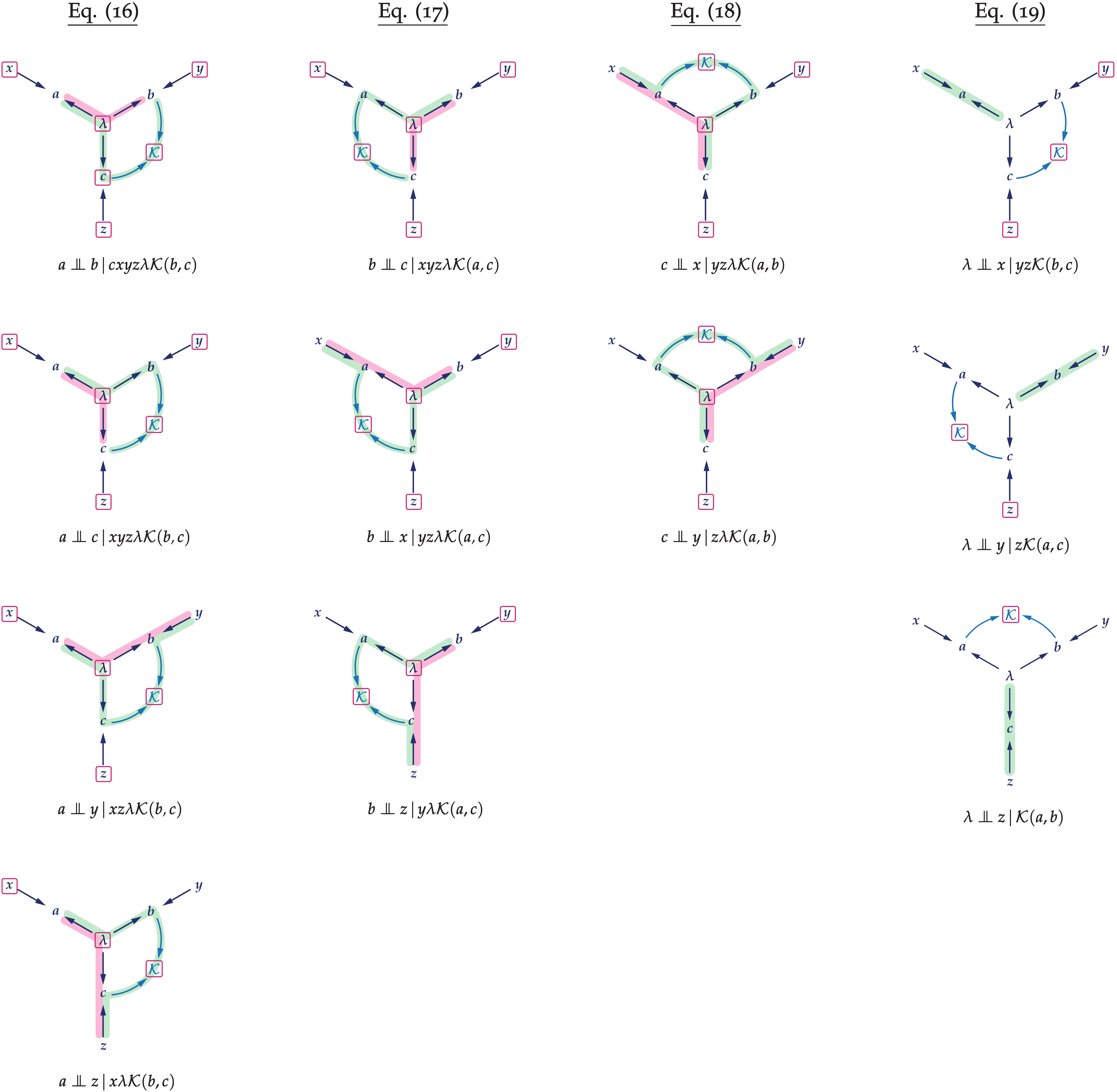}
\caption{\label{Fig_Proof_ABC}{\bf\textsf{\mbox{Graphical proof of \textbf{Theorem~\ref{theorem}} (\textit{three} parties).}}} Each graph illustrates structure of conditioning in Eqs.~(\ref{locality-proof-three-1})\,-\,(\ref{free-choice-proof-three}) depicted by red boxes around the variables. Marked in purple and green are paths joining variables of interest for which the respective independencies are inferred by the $d$-separation \textbf{\textsf{Rules~\ref{1}\,-\,\ref{3}}}. In case of Eqs.~(\ref{locality-proof-three-1})\,-\,(\ref{locality-proof-three-3}) there are always two paths, and all of them are blocked by conditioning on the non-collider $\lambda$ (\textbf{\textsf{Rule~\ref{2}}}). As for  Eq.~(\ref{free-choice-proof-three}), in each case the only path is blocked by the respective collider $a$, $b$ and $c$ (\textbf{\textsf{Rule~\ref{1}}}). Note that dropping one of the arrows in each of the graphs, as allowed by condition in Eq.~(\ref{all-but-one-EQ-three-SI}) and explained in Fig.~\ref{Fig_CausalStructure_ABC}, is not accidental (since otherwise it would open additional paths and spoil the independence pattern).}
\end{figure}

\newpage

\section*{$\bullet$ \,Proof of \textbf{Theorem~\ref{theorem}} for \textit{any} number of parties}

\textit{[The proof follows the lines of reasoning for the case of two and three parties.]}

\begin{proof} We consider the general case of $N$ parties in a Bell experiment with post-selection conforming to the \textit{all-but-one} principle in \textbf{\textsf{Definition~\ref{all-but-one}}}, i.e we have 
\begin{eqnarray}\label{all-but-one-EQ-SI}
\mathcal{K}\,=\,\mathcal{K}(a_{\scriptscriptstyle 1},...\,,\!\not\!{a}_{\scriptscriptstyle k},...\,,a_{\scriptscriptstyle N})\qquad\text{for each \ $k=1,...\,,N$}\,,\quad
\end{eqnarray}
where $\!\not\!{a}_{\scriptscriptstyle k}$ means that $k$-th variable is missing from the list of all outcomes. The causal structure is then given by the graph $\mathcal{G}_{\mathcal{K}}$ in Fig.~\ref{Fig_CausalStructure_N}. Note that condition Eq.~(\ref{all-but-one-EQ-SI}) entails that one of the $N$ arrows coming to $\mathcal{K}$ can be always erased without changing the generated statistics (cf. Fig.~\ref{Fig_CausalStructure_ABC}). In particular, this means that conditional independencies inferred by dropping one of those arrows remain the same.
\begin{figure}[h]
\centering
\includegraphics[width=0.8\columnwidth]{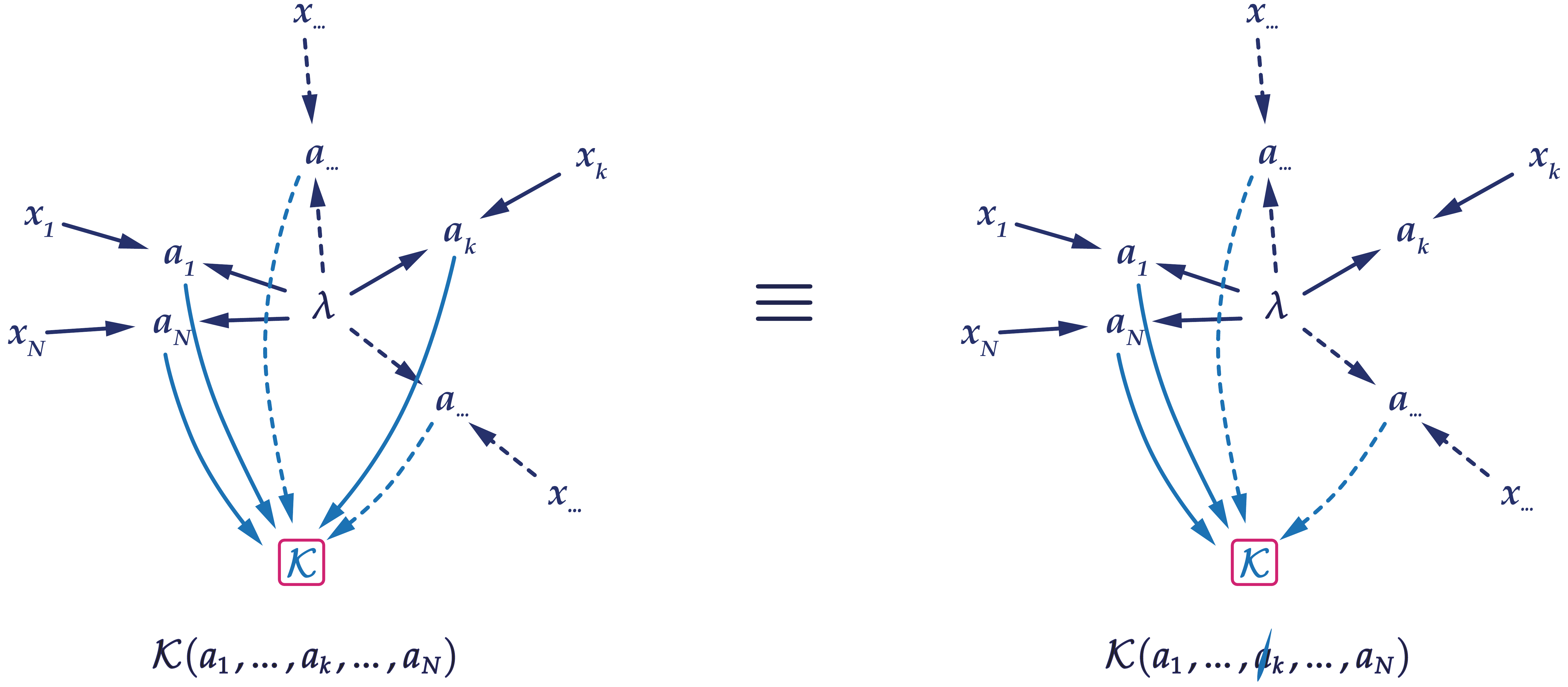}
\caption{\label{Fig_CausalStructure_N}{\bf\textsf{\mbox{Causal structure in a Bell experiment for \textit{any} number parties with post-selection.}}} On the left, unfolded on a plain is the causal graph $\mathcal{G}_{\scriptscriptstyle \mathcal{K}}$ describing  relations between variables in a Bell experiment ($a_{\scriptscriptstyle 1},...\,,a_{\scriptscriptstyle N}$ - outcomes, $x_{\scriptscriptstyle 1},...\,,x_{\scriptscriptstyle N}$ - measurement settings, and $\lambda$ - hidden variable). Variable $\mathcal{K}$ represents post-selection (where the red box means conditioning).  Cf. Fig.~\ref{Fig_CausalStructure}.
On the right, the causal graph with one arrow $a_{\scriptscriptstyle k}\rightarrow\mathcal{K}$ erased which is equivalent to graph $\mathcal{G}_{\scriptscriptstyle \mathcal{K}}$, if the \textit{all-but-one} principle in Eq.~(\ref{all-but-one-EQ-SI}) holds.}
\end{figure}

The proof of \textbf{\textsf{Theorem~\ref{theorem}}} consists of justifying that both Eqs.~(\ref{locality-K}) and (\ref{free-choice-K}) in \textbf{\textsf{Definition~\ref{Safe-post-selection}}} hold under the condition in Eq.~(\ref{all-but-one-EQ-SI}). 

By the repeated use of the chain rule, we get
\begin{eqnarray}
P_{a_1...\,a_N|x_1...\,x_N\lambda\mathcal{K}}\ =\ \prod_{\scriptscriptstyle k\,=\,1}^{\scriptscriptstyle N}\ P_{a_k|a_{k+1}...\,a_Nx_1...\,x_N\lambda\mathcal{K}}\,.
\end{eqnarray}
In order to prove Eq.~(\ref{locality-K}) it remains to check that for each $k=1,...\, ,N$ the following sequence of conditional independencies holds:
\begin{eqnarray}\label{locality-proof-any-1}
\!\!\!\begin{array}{llllll}
P_{a_k|a_{k+1}...\,a_N\,x_1...\,x_N\lambda\mathcal{K}}&=&P_{a_k|a_{k+2}...\,a_N\,x_1...\,x_N\lambda\mathcal{K}}
&&\vspace{0.3cm}\\
&&\quad\quad\text{since}\quad a_{\scriptscriptstyle k}\indep  a_{\scriptscriptstyle k+1}\,|\,a_{\scriptscriptstyle k+2}...\,a_{\scriptscriptstyle N}\,x_{\scriptscriptstyle 1}...\,x_{\scriptscriptstyle N}\lambda\,\mathcal{K}(a_{\scriptscriptstyle1},...\,,\!\not\!{a}_{\scriptscriptstyle k},...\,,a_{\scriptscriptstyle N})\\
&=&\dots\vspace{0.2cm}\\
&=&P_{a_k|a_{l+1}...\,a_N\,x_1...\,x_N\lambda\mathcal{K}}
&&\vspace{0.3cm}\\
&&\quad\quad\text{since}\quad a_{\scriptscriptstyle k}\indep  a_{\scriptscriptstyle l}\,|\,a_{\scriptscriptstyle l+1}...\,a_{\scriptscriptstyle N}\,x_{\scriptscriptstyle1}...\,x_{\scriptscriptstyle N}\lambda\,\mathcal{K}(a_{\scriptscriptstyle1},...\,,\!\not\!{a}_{\scriptscriptstyle k},...\,,a_{\scriptscriptstyle N})\\
&=&\dots\vspace{0.2cm}\\
&=&P_{a_k|x_1...\,x_N\lambda\mathcal{K}}
&&\vspace{0.3cm}\\
&&\quad\quad\text{since}\quad a_{\scriptscriptstyle k}\indep  a_{\scriptscriptstyle N}\,|\,x_{\scriptscriptstyle1}...\,x_{\scriptscriptstyle N}\lambda\,\mathcal{K}(a_{\scriptscriptstyle1},...\,,\!\not\!{a}_{\scriptscriptstyle k},...\,,a_{\scriptscriptstyle N})
\,,
\end{array}
\end{eqnarray}
where $l>k$, and
\begin{eqnarray}\label{locality-proof-any-2}
\!\!\!\!\!\!\!\!\!\!\!\!\begin{array}{llllll}
P_{a_k|x_1...\,x_N\lambda\mathcal{K}}&=&P_{a_k|x_2...\,x_N\lambda\mathcal{K}}
&\quad\text{since}\quad a_{\scriptscriptstyle k}\indep  x_{\scriptscriptstyle 1}\,|\,x_{\scriptscriptstyle2}...\,x_{\scriptscriptstyle N}\lambda\,\mathcal{K}(a_{\scriptscriptstyle1},...\,,\!\not\!{a}_{\scriptscriptstyle k},...\,,a_{\scriptscriptstyle N})\vspace{0.2cm}\\
&=&\dots\vspace{0.2cm}\\
&=&P_{a_k|x_{m+1}...\,x_N\lambda\mathcal{K}}
&\quad\text{since}\quad a_{\scriptscriptstyle k}\indep  x_{\scriptscriptstyle m}\,|\,x_{\scriptscriptstyle m+1}...\,x_{\scriptscriptstyle N}\lambda\,\mathcal{K}(a_{\scriptscriptstyle1},...\,,\!\not\!{a}_{\scriptscriptstyle k},...\,,a_{\scriptscriptstyle N})\vspace{0.2cm}\\
&=&\dots\vspace{0.2cm}\\
&=&P_{a_k|x_{k}\,x_{l+1}...\,x_N\lambda\mathcal{K}}
&\quad\text{since}\quad a_{\scriptscriptstyle k}\indep  x_{\scriptscriptstyle l}\,|\,x_{\scriptscriptstyle k}\,x_{\scriptscriptstyle l+1}...\,x_{\scriptscriptstyle N}\lambda\,\mathcal{K}(a_{\scriptscriptstyle1},...\,,\!\not\!{a}_{\scriptscriptstyle k},...\,,a_{\scriptscriptstyle N})\vspace{0.2cm}\\
&=&\dots\vspace{0.2cm}\\
&=&P_{a_k|x_{\scriptscriptstyle k}\lambda\mathcal{K}}
&\quad\text{since}\quad a_{\scriptscriptstyle k}\indep  x_{\scriptscriptstyle N}\,|\,x_{\scriptscriptstyle k}\lambda\,\mathcal{K}(a_{\scriptscriptstyle1},...\,,\!\not\!{a}_{\scriptscriptstyle k},...\,,a_{\scriptscriptstyle N})
\,,
\end{array}
\end{eqnarray}
where $l>k>m$. These independencies can be justified by the $d$-separation tools applied to the causal graph $\mathcal{G}_{\scriptscriptstyle \mathcal{K}}$. It boils down to the inspection of all paths joining variables in question. See Fig.~\ref{Fig_Proof_N} (cf. Fig.~\ref{Fig_Proof_ABC}). In each case there are $N-1$ possible paths which are all blocked by conditioning on the non-collider node $\lambda$ (\textbf{\textsf{Rule~\ref{2}}}). Like before, we use all conditions in Eq.~(\ref{all-but-one-EQ-SI}) to get the results (as readily seen from Fig.~\ref{Fig_Proof_N}, where in each case a different arrow $a_{\scriptscriptstyle k}\rightarrow\mathcal{K}$ is missing).

Similarly, the proof of Eq.~(\ref{free-choice-K}) boils down to the following conditional independencies:
\begin{eqnarray}\label{free-choice-proof-any}
\begin{array}{llllll}
P_{\lambda|x_1...\,x_N\mathcal{K}}&=&P_{\lambda|x_2...\,x_N\mathcal{K}}
&&\quad\text{since}\quad \lambda\indep  x_{\scriptscriptstyle 1}\,|\,x_{\scriptscriptstyle2}...\,x_{\scriptscriptstyle N}\mathcal{K}(\!\not\!{a}_{\scriptscriptstyle1},a_{\scriptscriptstyle 2},...\,,a_{\scriptscriptstyle N})\vspace{0.2cm}\\
&=&\dots\vspace{0.2cm}\\
&=&P_{\lambda|x_{k+1}...\,x_N\mathcal{K}}
&&\quad\text{since}\quad \lambda\indep  x_{\scriptscriptstyle k}\,|\,x_{\scriptscriptstyle k+1}...\,x_{\scriptscriptstyle N}\mathcal{K}(a_{\scriptscriptstyle1},...\,,\!\not\!{a}_{\scriptscriptstyle k},...\,,a_{\scriptscriptstyle N})\vspace{0.2cm}\\
&=&\dots\vspace{0.2cm}\\
&=&P_{\lambda|\mathcal{K}}
&&\quad\text{since}\quad \lambda\indep  x_{\scriptscriptstyle N}\,|\,\mathcal{K}(a_{\scriptscriptstyle1},...\,,a_{\scriptscriptstyle N-1},\!\not\!{a}_{\scriptscriptstyle N})
\,.
\end{array}
\end{eqnarray}
Here there is only one path joining the relevant variables that is blocked by the respective collider $a_{\scriptscriptstyle k}$ (\textbf{\textsf{Rule~\ref{1}}}). See Fig.~\ref{Fig_Proof_N} (cf. Fig.~\ref{Fig_Proof_ABC}). Note that the lacking arrow $a_{\scriptscriptstyle k}\rightarrow\mathcal{K}$ is crucial, since otherwise it would unblock the respective collider $a_{\scriptscriptstyle k}$ by conditioning on its descendent $\mathcal{K}$ (\textbf{\textsf{Rule~\ref{3}}}).

\end{proof}

\begin{figure*}[h]
\centering
\includegraphics[width=1\columnwidth]{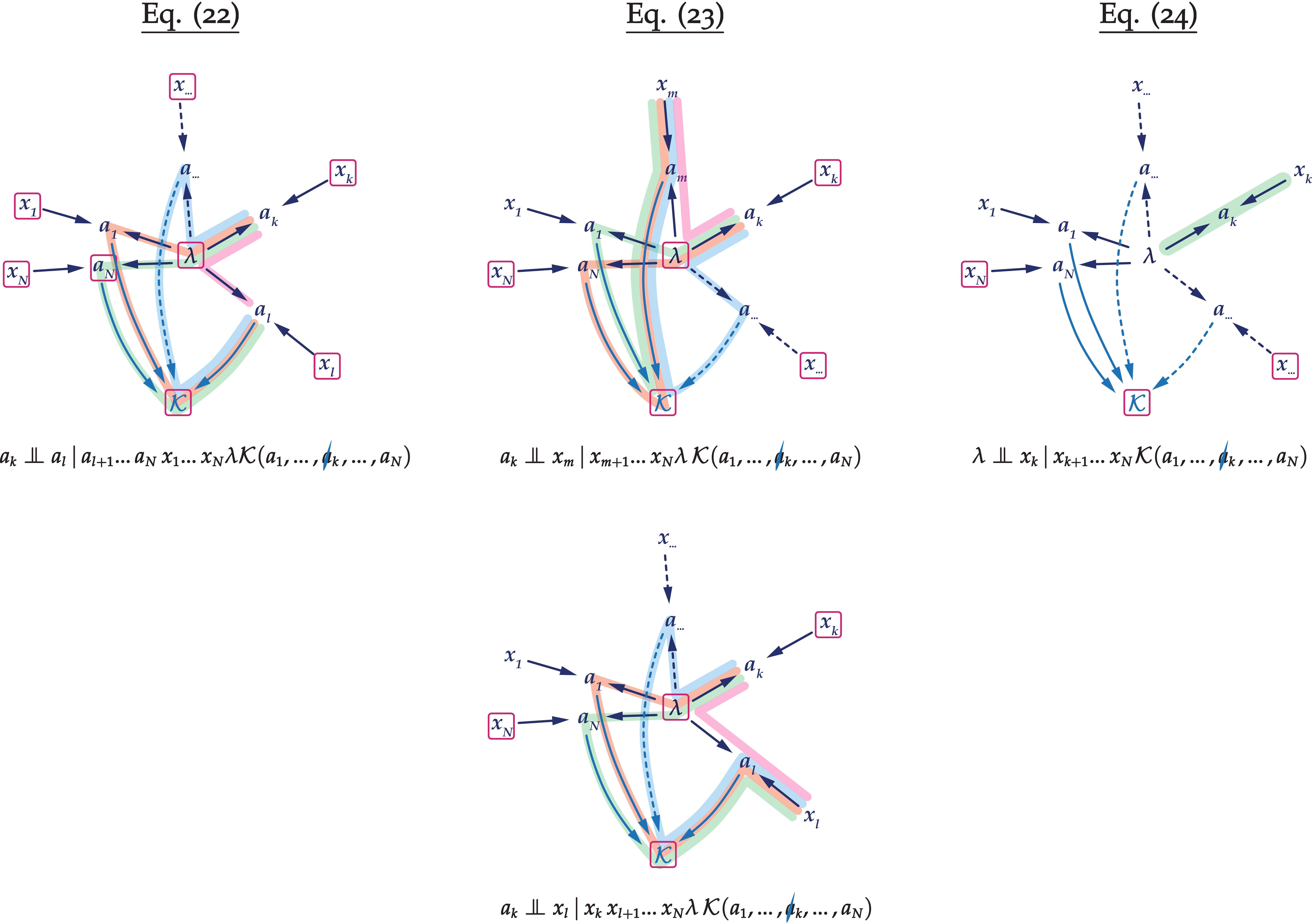}
\caption{\label{Fig_Proof_N}{\bf\textsf{\mbox{Graphical proof of \textbf{\textsf{Theorem~\ref{theorem}}} (\textit{any} number of parties).}}} We give three generic cases used for the justification of conditional independencies in Eqs.~(\ref{locality-proof-any-1})\,-\,(\ref{free-choice-proof-any}). Red boxes around the variables depict the structure of conditioning. Marked in purple, green, orange and blue  are paths joining variables of interest, $a_{\scriptscriptstyle k}$ and $a_{\scriptscriptstyle l}$ (resp. $x_{\scriptscriptstyle l}$), for which the respective independencies are inferred by the $d$-separation \textbf{\textsf{Rules~\ref{1}\,-\,\ref{3}}}. In case of Eqs.~(\ref{locality-proof-any-1}) and (\ref{locality-proof-any-2}) there are $N-1$ paths, and all of them are blocked by conditioning on the non-collider $\lambda$ (\textbf{\textsf{Rule~\ref{2}}}). As for Eq.~(\ref{free-choice-proof-any}), the only path joining $\lambda$ and $x_{\scriptscriptstyle k}$ is blocked by the collider $a_{\scriptscriptstyle k}$ (\textbf{\textsf{Rule~\ref{1}}}). Let us note that the lack of the respective arrow $a_{\scriptscriptstyle k}\rightarrow\mathcal{K}$ in each of the graphs is essential (in order not to introduce unwanted paths in case of Eqs.~(\ref{locality-proof-any-1}) and (\ref{locality-proof-any-2}), or to prevent lifting the block from the collider  by \textbf{\textsf{Rule~\ref{3}}} in case of Eq.~(\ref{free-choice-proof-any})).}
\end{figure*}

\end{document}